\def\arxiv{0}
\def\newarxiv{1} 
\ifnum\arxiv=1
    \documentclass[11pt]{article}     \usepackage[letterpaper,hmargin=1.0in,vmargin=1.0in]{geometry}
    \usepackage[numbers]{natbib}
\else
    \documentclass[a4paper,UKenglish,cleveref, autoref, thm-restate]{lipics-v2021}
    \ifnum\newarxiv=1
        \hideLIPIcs 
    \fi
\fi

\usepackage{amsmath, amsthm, amsfonts}
\usepackage[utf8]{inputenc}
\usepackage{accents}
\usepackage{setspace}
\usepackage{hyperref}
\usepackage{comment}
\usepackage{appendix}
\usepackage{graphicx}
\usepackage{multirow}
\usepackage{ragged2e}
\usepackage{caption}
\usepackage{subcaption}
\usepackage[normalem]{ulem}
\usepackage{thmtools}
\usepackage{thm-restate}
\usepackage{MnSymbol}
\usepackage{xcolor}
\usepackage[linesnumbered,ruled,vlined]{algorithm2e}
\usepackage{algpseudocode}
\usepackage{paralist}
\usepackage{framed}
\usepackage{array}
\usepackage[framemethod=tikz]{mdframed}
\usepackage{xspace}
\usepackage{url}

\SetCommentSty{mycommfont}
\SetKwInput{KwInput}{Input}                
\SetKwInput{KwOutput}{Output}              
\SetKwInput{KwInit}{Invariant}
\SetKw{Continue}{continue}
\def\elmerge{{\sf elementary merges}}
\def\conmerge{{\sf concatenation merges}}
\def\cymerge{{\sf cycle merges}}
\def\elmergee{{\sf elementary merge}}
\def\conmergee{{\sf concatenation merge}}
\def\cymergee{{\sf cycle merge}}

\def\heads{{\tt heads}}
\def\tails{{\tt tails}}

\ifnum\arxiv=1
    \newtheorem{theorem}{Theorem}
    \newtheorem{corollary}{Corollary}
    \newtheorem{lemma}{Lemma}
    \newtheorem{definition}{Definition}

    \newtheorem{remark}{Remark}
    \newtheorem{claim}{Claim}
    \providecommand{\keywords}[1]
    {
      \small	
      \textbf{\textit{Keywords---}} #1
    }
\fi

\newcommand{\Exp}{\mathbb{E}}
\newcommand{\E}{\Exp}
\newcommand{\poly}{{\rm poly}}

\newcommand{\eqdef}{\stackrel{\rm def}{=}}

\def\congest{{\sf CONGEST}}
\def\pram{{\sf PRAM}}
\def\cpram{{\sf CREW-PRAM}}
\def\ccpram{{\sf CRCW-PRAM}}
\def\epram{{\sf EREW-PRAM}}

\def\dirac{{\sc Dirac}}
\def\ore{{\sc Ore}}
\def\rk{{\sc RK}}

\ifnum\arxiv=0
    \title{Distributed \congest\ Algorithm for Finding Hamiltonian Paths in Dirac Graphs and Generalizations}
    \titlerunning{Dist. \congest\ Alg. for Finding Ham. Paths in Dirac Graphs and Generalizations} 
\fi
\ifnum\arxiv=1
    \author{Noy Biton\thanks{Efi Arazi School of Computer Science, Reichman University, Israel. Email: {\tt Noy.Biton@post.idc.ac.il}} \and Reut Levi\thanks{Efi Arazi School of Computer Science, Reichman University, Israel. Email: {\tt reut.levi1@idc.ac.il}} \and Moti Medina\thanks{Faculty of Engineering, Bar-Ilan University, Ramat Gan, Israel. Email: {\tt moti.medina@biu.ac.il}}}
    \date{}
\else
    \author{Noy Biton}{Efi Arazi School of Computer Science, Reichman University, Israel}{Noy.Biton@post.idc.ac.il}{}{The author was supported by the Israel Science Foundation under Grant 1867/20.}

    \author{Reut Levi}{Efi Arazi School of Computer Science, Reichman University, Israel}{reut.levi1@idc.ac.il}{https://orcid.org/0000-0003-3167-1766}{The author was supported by the Israel Science Foundation under Grant 1867/20.}

    \author{Moti Medina}{Faculty of Engineering, Bar-Ilan University, Ramat Gan, Israel}{moti.medina@biu.ac.il}{https://orcid.org/0000-0002-5572-3754}{The author was supported by the Israel Science Foundation under Grant 867/19.}

    \authorrunning{N. Biton, R.Levi and M. Medina}

    \Copyright{Noy Biton, Reut Levi and Moti Medina}
    


    \ifnum\newarxiv=0
        \relatedversion{A full version of the paper is available at  \url{https://arxiv.org/abs/2302.00742}.}
    \fi
    \ifnum\newarxiv=1
        \nolinenumbers 
    \fi
    
    \EventEditors{John Q. Open and Joan R. Access}
    \EventNoEds{2}
    \EventLongTitle{42nd Conference on Very Important Topics (CVIT 2016)}
    \EventShortTitle{CVIT 2016}
    \EventAcronym{CVIT}
    \EventYear{2016}
    \EventDate{December 24--27, 2016}
    \EventLocation{Little Whinging, United Kingdom}
    \EventLogo{}
    \SeriesVolume{42}
    \ArticleNo{23}
\fi
\begin{document}

\ifnum\arxiv=1
\title{Distributed \congest\ Algorithm for Finding Hamiltonian Paths in Dirac Graphs and Generalizations}
\fi
\maketitle
\begin{abstract}
    We study the problem of finding a Hamiltonian  cycle under the promise that the input graph has a minimum degree of at least $n/2$, where $n$ denotes the number of vertices in the graph. The classical theorem of Dirac states that such graphs (a.k.a. Dirac graphs) are Hamiltonian, i.e., contain a Hamiltonian cycle.  Moreover, finding a Hamiltonian  cycle in Dirac graphs can be done in polynomial time in the classical centralized model.

    This paper presents a randomized distributed \congest\ algorithm that finds w.h.p. a Hamiltonian cycle (as well as maximum matching) within $O(\log n)$ rounds under the promise that the input graph is a Dirac graph. This upper bound is in contrast to general graphs in which both the decision and search variants of Hamiltonicity require $\tilde{\Omega}(n^2)$ rounds, as shown by  Bachrach et al. [PODC'19].

    In addition, we consider two generalizations of Dirac graphs: Ore graphs and Rahman-Kaykobad graphs [IPL'05]. In Ore graphs, the sum of the degrees of every pair of non-adjacent vertices is at least $n$, and in Rahman-Kaykobad graphs, the sum of the degrees of every pair of non-adjacent vertices plus their distance is at least $n+1$. We show how our algorithm for Dirac graphs can be adapted to work for these more general families of graphs.
\end{abstract}
\if\arxiv=1
\fi
\ifnum\arxiv=0
    \maketitle
\fi
\section{Introduction}
The Hamiltonian path {\color{black} and Hamiltonian cycle }problems are fundamental in computer science and appeared in Karp's 21 NP-complete problems~\cite{K72}. A Hamiltonian path is a path that visits every vertex in the graph exactly once {\color{black}and a Hamiltonian cycle is a cycle that visits every vertex in the graph exactly once. We say that a graph is \emph{Hamiltonian} if it contains a Hamiltonian cycle.}

While the problem is hard in general (assuming ${\rm P \neq NP}$), for some classes of graphs, it is guaranteed that all their members are Hamiltonian. 
In particular, the classical theorem of Dirac states that every graph in which the minimum degree is at least $n/2$ is Hamiltonian (we refer to graphs that satisfy this condition as \dirac\ graphs). This condition is tight in the sense that if we are only guaranteed that the minimum degree is at least $\alpha n$ for any $0 < \alpha < 1/2$, then the problem remains {\rm NP}-complete~\cite{DHK93}.

In FOCS’87 Goldberg proposed the question of whether there is an NC-algorithm for finding a Hamiltonian cycle in \dirac\ graphs.~\footnote{See~\cite{DHK93} and~\cite{sarkozy2009fast} for more details.}
This question was answered affirmatively by Dahlhaus et al.~\cite{dahlhaus1988optimal, DHK93}, who gave a fast parallel algorithm on \cpram\ to find a Hamiltonian cycle in Dirac graphs. Their algorithm works in $O(\log^4n)$ parallel time and uses $O(n+m)$ number of processors where $m$ denotes the number of edges of the graphs. 

Aside from the theoretical appeal of the problem, finding a Hamiltonian cycle in a graph also provides us with a maximum matching of the graph (and even perfect matching when $n$ is even). 
Therefore, it is natural to ask whether the algorithm of Dahlhaus et al.~\cite{DHK93} can be translated to the \congest\ model. 
To this end, one may attempt to use the \pram\ simulation of Lotker, Patt-Shamir, and Peleg~\cite{lotker2006distributed} for diameter-$2$ graphs (which applies for \dirac\ graphs).
However, this simulation is only valid when the number of processors is linear in the number of vertices in the graph.  
Another attempt is to use the more general transformation of Ghaffari and Li~\cite{GJ18} that provide a randomized \congest\ algorithm that simulates any \ccpram\ algorithm that uses $2m$ processors, runs in $T$ parallel rounds, and operates on the input graph $G$ that is stored in the \pram's shared memory. 
The round complexity of the attained \congest\  algorithm is $T\cdot \tau_{\rm mix}(G)\cdot 2^{O(\sqrt{\log n})}$, where $\tau_{\rm mix}(G)$ is the mixing-time of $G$. Thus even for constant mixing time, this yields a simulation of the algorithm of~\cite{DHK93} in \congest\ with $2^{O(\sqrt{\log n})}$ rounds. Moreover, since the mixing time of Dirac's graph can be $\Theta(n)$~\footnote{Consider a Dirac graph, over $n$ vertices, which is composed of two cliques of size $n/2$ with a perfect matching between the cliques.}, the round complexity of this simulation can be super-linear in $n$.
Consequently, our goal is to improve upon this round complexity by directly designing an algorithm for the \congest\ model. Indeed we provide an algorithm with an exponential improvement in the round complexity. Specifically, our algorithm for finding a Hamiltonian cycle in Dirac graphs runs in $O(\log n)$ rounds. When the algorithm terminates, each vertex outputs the identifier of the vertex that is before it and the vertex that is after it on the cycle~\footnote{We note that although the input graph is undirected, the outputs of the vertices yield an oriented Hamiltonian cycle (or path).}.  

In the \congest\ model, it is standard to assume that the processors have unbounded computational power. Therefore, one may wonder whether finding a Hamiltonian cycle in general graphs in $o(n^2)$ rounds is possible.  
It was recently shown by Bachrach et al.~\cite{BCHDELP19} that even the problem of testing Hamiltonicity in the \congest\  model~\cite{P00} requires $\tilde{\Omega}(n^2)$ rounds.
Therefore, it is natural to focus on restricted families of graphs such as \dirac\ graphs and their generalizations.  
%
%
%
Since the classical result of Dirac, there have been many generalizations of Dirac's theorem (see~\cite{H13} and references therein), e.g., graph families that are defined by sufficient conditions on degrees, neighborhoods, and other graph parameters~\cite{O60,bermond1976hamiltonian,bondy1971large,LINIAL1976297,bondy1976method,bondy1980longest,fournier1985conjecture,flandrin1991hamiltonism}. The first important generalization of Dirac's theorem is by Ore~\cite{O60} who proved that every graph in which the sum of degrees of each pair of non-adjacent vertices is at least $n$ is Hamiltonian. 
A more recent generalization, which also generalizes Ore's theorem, and allows graphs with less edges, is by Rahman and Kaykobad~\cite{RK05} who proved that every graph in which the sum of degrees of each pair of non-adjacent vertices plus their distance is at least $n+1$ has a Hamiltonian path. 
We refer to graphs that satisfy these conditions as \ore\ graphs and \rk\ graphs, respectively. 
We prove that our distributed \congest\ algorithm and its analysis can be adapted (without changing the round complexity asymptotically) for these generalizations of \dirac\ graphs as well.  
\paragraph*{Our Results. }
Our main result is stated in the following theorem. 
\begin{restatable}{theorem}{mainham}
\label{thm:ham}
There exists a distributed algorithm for computing a Hamiltonian {\color{black} cycle} in \dirac\ graphs that runs in
$O(\log n)$ rounds in the \congest\ model. The algorithm succeeds with high probability.~\footnote{We say that an event occurs \emph{with high probability (w.h.p.)} if it occurs with probability at least $1-1/\poly(n)$.} 
\end{restatable}

We also prove the following, more general, theorem in Appendix~\ref{sec:rk}.
\begin{restatable}{theorem}{mainhamtwo}
\label{thm:ham2}
There exists a distributed algorithm for computing a Hamiltonian {\color{black} cycle} in \ore\ graphs {\color{black} and a Hamiltonian path in}  \rk\ graphs that runs in
$O(\log n)$ rounds in the \congest\ model. The algorithm succeeds with high probability.
\end{restatable}

A nice outcome of Theorem~\ref{thm:ham2} is that a sequential simulation of the \congest\ algorithm for the \rk\ graphs yields a polynomial time sequential algorithm for finding a Hamiltonian path in these graphs as well.

\subsection{High-level Description of the Algorithm}
Our algorithm maintains a path-cover of the graph, where a path-cover is a set of paths in the graph such that each vertex of the graph belongs to exactly one of the paths in the set.

Initially, the path-cover consists of paths of constant length. Hence the size of the initial path-cover is linear in $n$. Then the algorithm proceeds in iterations, where in each iteration, the size of the path-cover decreases by a constant factor, with constant probability. Consequently, after $\Theta(\log n)$ iteration, the size of the path-cover is $1$. Namely, a Hamiltonian path is found.
The decrease in the size of the path-cover occurs as in each iteration (with constant probability) a constant fraction of the paths are merged into other paths, as we describe next.

We consider three types of merges. An \emph{elementary merge} occurs when a path $P$ is merged into a path $Q$ by connecting the endpoints of $P$ to the endpoints of an edge of $Q$.
A \emph{concatenation merge} occurs when two paths are merged by connecting their endpoints. 
Finally, a \emph{cycle merge} occurs when merging two cycles that are connected with an edge into a single path.

In each iteration, the algorithm proceeds as follows. First, the paths of the path-cover are paired. 
Now, two phases are performed, as follows.  
In the first phase of each iteration, only pairs of paths for which a special condition (which we describe momentarily) holds are merged. The special condition guarantees that each one of these pairs can be merged into a single path. Specifically, a pair of paths, $(P, Q)$ satisfy this special condition if the subgraph induced on the vertices of each one of them has a Hamiltonian cycle and additionally $P$ and $Q$ are connected with an edge to each other. 

In the second phase of each iteration, concatenations and elementary merges are performed. The merges that are performed in this phase are selected as follows. Let $P$ be a path in the path cover. The {\em elements} of $P$ consists of its edges and its endpoints. 
Initially, each element of $P$ reserves itself to at most a single path, $Q$. The path $Q$ is selected as follows. If the (reserved) element is an endpoint of $P$, $v$, then $Q$ is a path with an endpoint incident to $v$, chosen uniformly at random from the set of such paths.
Otherwise, if the (reserved) element is an edge~\footnote{We note that although the graph is undirected, we keep an orientation on the edges of the paths of the path cover (so the obtained paths are directed).
} of $P$, $(u, v)$, then $Q$ is a path with an endpoint incident to $u$, chosen uniformly at random from the set of such paths.
A  reservation of an element of $P$ to a path $Q$ grants $Q$ the exclusive right to be merged into $P$ using this element.
The purpose of these reservations is to avoid a scenario in which two different paths are trying to merge into another path by using the same element.

We say that a reservation of an element is \emph{useful} for a path $Q$ if $Q$ can be merged into another path via this element. By construction, all reservations of endpoints are useful. However, a reservation of an edge may be non-useful since the decision to reserve an edge to a path, $Q$, is done only by one of the endpoints of $e$. Hence, it might be the case that the other endpoint of $e$ is not adjacent to the other endpoint of $Q$.
Nonetheless, we show that on expectation, a constant fraction of the paths will have at least one useful reservation. 
After setting the reservations, each path is notified of the set of its useful reservations (if any). This can be carried out without congestion because of the exclusivity of the reservations. Then each path arbitrarily selects one of its useful reservations.

At this point, each path that has a useful reservation can be merged into another path via the respective reserved element exclusively. However, there are two problems with executing these merges. The first problem is that we want to avoid lengthy sequences of merges as this blows up the round complexity of the algorithm. Roughly speaking, this comes from the fact that when we merge paths (possibly many) into a single new path, all vertices of the corresponding (old) paths are updated about the identity of the new path. Moreover, for each new path (which is an outcome of possibly many merging operations), the algorithm constructs a spanning tree, of depth $2$, which spans the vertices of the path. See more details in Appendix~\ref{sec:imp} on the role of these spanning trees in our algorithm.

The second problem is that these sequences of merges may be conflicting.
A simple example of a conflict is when a graph $P$ tries to merge into $Q$ via an edge of $Q$, and $Q$ tries to merge into $P$ via an edge of $P$. Clearly, these two merges cannot  be carried out simultaneously. Moreover, this example can be extended into arbitrarily long cycles.   

Fortunately, these two problems can be remedied by the following simple idea. Each path tosses a fair coin. Then each path, $P$, is merged via its selected (reserved) element only if $P$ tossed \tails{} and $Q$ tossed \heads{}, where $Q$ is the path of the respective element. On expectation $1/4$ of the merges will be in the ``right'' orientation. In our analysis, we prove that this suffices for our needs (see more details on Subsection~\ref{sec:progress}). 
This completes the description of the second phase of each iteration and concludes the description of the algorithm. 

\subsection{Correctness and Analysis of the Algorithm}
The analysis of the algorithm has two ingredients.
The first and main ingredient is showing that for any fixed iteration, the size of the path cover decreases by a constant fraction on expectation.
The second ingredient is showing that after $\Theta(\log n)$ iterations, with high probability, all paths are merged into a single  Hamiltonian path. 
   
For the first ingredient of the analysis, we define the notion of being a {\em  good path}. Roughly speaking, a path is good if there are sufficiently many edges connecting its endpoints to vertices of other paths (the actual definition is more cumbersome than this, but this is essentially the property that we need). This property guarantees that any good path has many options for merging into other paths. In particular, the number of options is linear in the size of the current path cover (Lemma~\ref{lemma:leftdeg}). We use this fact to show that, on any fixed iteration, a good path receives a useful reservation with constant probability (Claim~\ref{clm:single}).

Finally, we prove that there are sufficiently many good paths. Recall that at the beginning of each iteration, the paths are paired. 
In the algorithm analysis, we prove that for each one of the pairs of paths, $(P, Q)$, that were not merged in the first phase, either $P$ or $Q$ are good with respect to the current path-cover. We then show that consequently, this guarantees that with constant probability, a constant fraction of these paths will be merged into other paths in the second phase of the algorithm. 

\subsection{Adaptation of the Algorithm for \ore\ and \rk\ Graphs}

Since the family of \rk\ graphs contains \ore\ graphs we, from now on, focus on \rk\ graphs.
We begin by proving some structural properties of \rk\ graphs. One of these properties is that the vertices in \rk\ graphs can be partitioned into $3$ sets $A$, $C$, and $H$ where all the vertices in $H$ satisfy Dirac's condition and the subgraph induced on each one of the sets $A$ and $C$, form a clique. 
This structure allows us to perform as in the algorithm for \dirac\ graphs with the only difference that at the beginning of each iteration, we get rid of (almost) all the paths in which one endpoint is not in $H$. 
Another technicality that we need to handle is the fact that our algorithm for \dirac\ graphs uses spanning trees to manage the communication within the different paths in the path cover. In \rk\ graphs these spanning trees may not span the entire respective paths. However, we show that with a slight adaptation, it is possible to maintain communication within the paths while adding only a constant factor blow-up in the round complexity.   
%

\subsection{Related Work}
\ifnum\arxiv=1
\paragraph*{Parallel Algorithms. }
As mentioned above, Dahlhaus et al.~\cite{DHK93} gave a $O(\log^4n)$ \cpram\ algorithm that uses a linear number of processors to find a Hamiltonian cycle in Dirac graphs. Another generalization of Dirac graphs are \emph{Chv\'{a}tal graphs}. A graph is called a Chv\'{a}tal graph if its degree sequence $d_1 \leq d_2 \leq \cdots \leq d_n$ satisfies that for every $k < n/2$, $d_k \leq k$ implies that $d_{n-k} \geq n-k$. Chv\'{a}tal proved (see e.g.,~\cite{bollobas2004extremal}) that Chv\'{a}tal graphs are also Hamiltonian. In~\cite{bondy1976method} a sequential polynomial time algorithm that finds Hamiltonian cycles in Chv\'{a}tal graphs has been shown. S\'{a}rk\"{o}zy~\cite{sarkozy2009fast} proved that a deterministic $O(\log^4 n)$ time \epram\  algorithm with a polynomial number of processors that finds a Hamiltonian cycle in a \emph{$\eta$-Chv\'{a}tal graphs} exists. A graph is called $\eta$-Chv\'{a}tal graphs if for every $k < n/2$, $d_k \leq \min\{k+\eta n,n/2\}$, it holds that $d_{n-k-\eta n} \geq n-k$, where $0 < \eta < 1$.
\fi

\paragraph*{Distributed Algorithms with a Promise. }
It is known that a random $G(n, p)$ graph contains w.h.p. a Hamiltonian cycle if $p$ is at least $(\log n + \log\log n + t(n))/n$, for any divergent function $t(n)$~\cite{Bol11}.
Thus for any such $p$, the problem of deciding whether the graph is Hamiltonian becomes trivial; however, finding the Hamiltonian path or cycle, in this case, is still non-trivial. 
The problem of finding a Hamiltonian cycle in the distributed setting was initially studied by Levy et al.~\cite{LLP05} that provided an upper bound whose round complexity is $O(n^{3/4+\epsilon} )$ that w.h.p. finds a Hamiltonian cycle given that $p = \omega(\sqrt{\log n}/n^{1/4})$. 
Thereafter, other upper bounds were designed in the \congest\ model (in which the message size is bounded) for dealing with various ranges of $p$, as described next. 
The algorithm of Chatterjee et al. ~\cite{CFPP18} works for $p \leq \frac{c \ln n}{n^\delta}, (0 < \delta \leq 1)$ and finds a Hamiltonian cycle w.h.p. in $\tilde{O}(n^\delta )$ rounds.  
For $p$ in $\tilde{\Omega}(1/\sqrt{n})$, Turau ~\cite{T20}, provides an algorithm that finds w.h.p. a Hamiltonian cycle in $O(\log n)$ rounds.
More recently, Ghaffari and Li~\cite{GJ18} showed the existence of a distributed algorithm for finding a Hamiltonian cycle, w.h.p., in $G(n,d)$ for $d=C\log n$ whose round complexity is $2^{O(\sqrt{\log n})}$, where $C$ is a sufficiently large constant.
\paragraph*{Lower Bounds in the \congest\ Model. } 
It is well known that certain properties can not be decided in the \congest\ model in a number of rounds which is $o(n^2)$~\cite{P00}. As mentioned above, Bachrach et al.~\cite{BCHDELP19} proved a lower bound of $\tilde{\Omega}(n^2)$ rounds for various problems in the \congest\ model, including  Testing Hamiltonicity in general graphs.

\ifnum\arxiv=0
    For other related work on Parallel Algorithms, see Appendix~\ref{app:related}.
\fi

\subsection{Comparison with the Algorithm of Dahlhaus et al.\texorpdfstring{~\cite{DHK93}}{[DHK93]}}

As mentioned above, our algorithm builds on ideas from the algorithm of Dahlhaus et al.~\cite{DHK93} for finding a Hamiltonian path in Dirac's graphs. 
Their algorithm also proceeds in iterations such that at each iteration it first performs cycle merges, then it performs concatenation merges and finally, it performs elementary merges.
However, the specific structure of their algorithm and how these merges are selected are quite different from our algorithm.

We shall demonstrate several structural differences without going into all the details of their algorithm (which is more involved than our algorithm). 
These differences serve us in obtaining an improved round complexity and a simpler algorithm.  
Thereafter we shall emphasize the specific differences that arise from the fact that our algorithm works in the \congest\ model rather than the \pram\ model.

The first structural difference is that on each iteration, before their algorithm turns into performing elementary merges it first has to exhaust most of the cycle merges, which requires an inner loop of $\Theta(\log n)$ steps and thereafter it exhausts most of the concatenation merges by executing a special subroutine of Israeli and Shiloach~\cite{israeli1986improved} which returns both a vertex cover and an approximated maximum matching.

This subroutine is executed twice. One time on the subgraph induced on the endpoints of the paths in the path cover and another time on an auxiliary graph where on one side we have the set of paths and on the other side we have the set of edges composing the paths. 

The reason that their algorithm exhausts the three types of merges in phases is that the progress of each phase relies on the exhaustion of the merges of the previous step.

For comparison, per iteration, our algorithm performs only one step of cycle merges and then one step in which both concatenation and elementary merges are performed simultaneously. 
We prove that this is sufficient to make enough progress per iteration. 

Another difference is that in their algorithm, at the beginning of each iteration, every path is classified into one of two types. We are able to avoid this classification altogether and use a somewhat different classification only in the analysis.  

We next list several challenges that arise specifically in the \congest\ model and in particular do not allow us to easily translate the algorithm of~\cite{DHK93} into the \congest\ model. 

The first problem is that we do not have shared memory among the processors so how can we efficiently merge even a pair of paths?
To this end, our algorithm maintains spanning trees, of depth $2$, on each one of the paths in the constructed path cover. 
Therefore, initially, we have a linear (in $n$) number of spanning trees (and this number decreases as the algorithm progresses). 
We show that each edge participates in at most $2$ different spanning trees and so communication within vertices of the same path can be maintained without causing congestion.  

Another problem is concerned with elementary merges. Consider a path cover $\mathcal{P}$ and an edge $(u, v)$ on a path $P\in \mathcal{P}$. 
An elementary merge of a path $Q$ into $P$ can be performed via $(u ,v)$ only if one endpoint of $Q$ is adjacent to $u$ and the other is adjacent to $v$. However, there might be many endpoints that are adjacent to $u$ or $v$, so how can we find the set of paths that can be merged via $(u, v)$ into $P$ without communicating too much between $u$ and $v$? As mentioned above, we show that we don't have to find this set. Specifically, we show that when $u$ reserves the edge $(u, v)$ to an endpoint that is adjacent to $u$, picked uniformly at random, then every path receives a useful reservation with constant probability. We remark that in this case, we rely on the randomness of our algorithm while the algorithm of~\cite{DHK93} is deterministic. 

Finally, we want to avoid long sequences of merges so we won't have to deal with long sequences of updates. 
To this end, we use the coin tosses of the vertices and perform merge operations only if they agree with the orientation defined by the coin tosses. As described above, this is also useful for avoiding conflicting merges.
Consequently, the merges can be carried out simultaneously with very little and local coordination.   

\section{Paths-Merge Types and Paths Classification}\label{sec:prelimdhk}
\paragraph*{Notation. }
Let $G=(V,E)$ be an undirected simple graph, where $V$ is the vertex set, and $E\subseteq \{\{u,v\} \mid u,v \in V\}$ is the edge set. Let $n$ denote $|V|$ and let $m$ denote $|E|$. 
For a path $P=(u_1, \ldots, u_{\ell})$, we denote by $V(P)$ the vertex set of $P$, i.e.,  $V(P) \triangleq\{u_1,\ldots, u_{\ell}\}$. 
For $v\in V$, let $N(v)$ denote the neighbors set of $v$ in $G$, that is $N(v)=\{u \in V \mid \{v,u\}\in E\}$. 
Let $d(v)$ denote the degree of $v$, i.e, $d(v) = |N(v)|$.
For a pair of vertices $u$ and $v$, let $\delta(u, v)$ denote the length of a shortest path between $u$ and $v$.  
We say that a set of paths in $G$, $\{P_i\}_{i=0}^{k-1}$, is a \emph{path-cover} of $G$ if its union covers the vertex set of $G$, that is, if $\bigcup_{i=1}^{k-1}V(P_i)=V(G)$.  
For a path $P$ in $G$, let $d_P(v)$ denote the number of neighbours of $v$ in $P$. 
For a path $P = (u, \ldots, v)$, we refer to the vertices $u$ and $v$ as the \emph{endpoints} of $P$.

\paragraph*{Path Merging Types. } Through the course of the algorithm's execution, the algorithm performs three kinds of path merging depending on the paths at hand: an elementary merge, a concatenation, and cycle merging (see Figure~\ref{fig:operations}). These merge operations are defined as follows.

\begin{definition}[Elementary merge~\cite{DHK93}]\label{def:elem}\sloppy
    Let $P=(u_1, \ldots, u_{\ell})$ and $Q=(v_1, \ldots, v_m)$ be two disjoint paths. If $\{u_1,v_i\}, \{u_{\ell},v_{i+1}\} \in E$, then $(v_1, \ldots, v_i, u_1, \ldots, u_{\ell}, v_{i+1}, \ldots, v_m)$ is a path. If $\{u_1,v_{i+1}\}, \{u_{\ell},v_{i}\} \in E$, then $(v_1, \ldots, v_i, u_{\ell}, \ldots, u_1, v_{i+1}, \ldots, v_m)$ is a path.  In either case, we say that we \emph{merged $P$ into $Q$ along the edge $\{v_i, v_{i+1}\}$}. We call this step an \emph{elementary merging operation}. 
\end{definition}
%
%
\begin{definition}[Concatenation~\cite{DHK93}]\label{def:concat}
    Let $P=(u_1, \ldots, u_{\ell})$ and $Q=(v_1, \ldots, v_m)$ be  two disjoint paths. If there is an edge connecting an  endpoint of $P$ (either $u_1$, or $u_{\ell}$) and an endpoint $v \in \{v_1, v_m\}$ of $Q$, then we can use any of these edges to concatenate $P$ and $Q$ and say that we \emph{concatenated $P$ to $Q$ along vertex $v$}.  We call this operation a \emph{concatenation}.
\end{definition}

%
\begin{definition}[Cycle Merging~\cite{DHK93}]\label{def:cyclemerge}
    Let $C$ and $D$ be two disjoint cycles. If there is an edge connecting a vertex from $C$ and a vertex from $D$, we can use this edge to get a path that passes through all the vertices of $C$ and $D$. We call this operation a \emph{cycle merging}.
\end{definition}
\begin{figure*}
    \centering
     \begin{subfigure}[b]{0.32\textwidth}
         \centering
         \includegraphics[width=\textwidth]{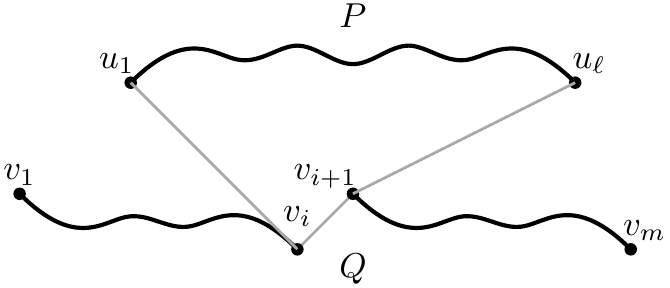}
         \caption{\elmergee{} operation}
         \label{fig:PathMerge}
     \end{subfigure}
          \hfill
     \begin{subfigure}[b]{0.32\textwidth}
         \centering
         \includegraphics[width=\textwidth]{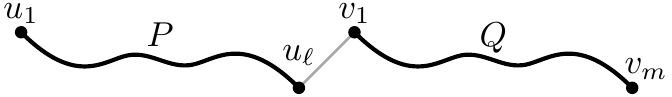}
         \caption{ \conmergee{} operation}
         \label{fig:PathCont}
     \end{subfigure}
     \hfill
     \begin{subfigure}[b]{0.32\textwidth}
         \centering
         \includegraphics[width=\textwidth]{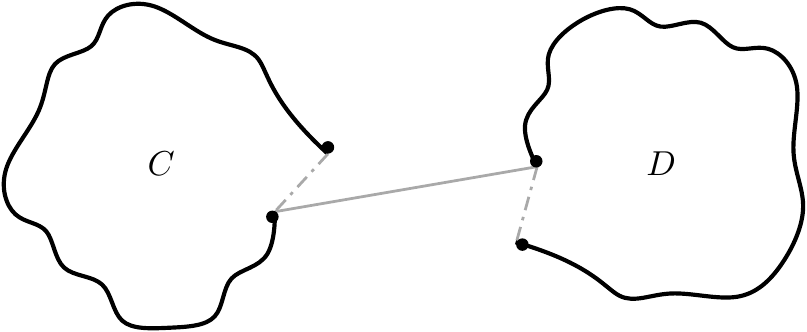}
         \caption{\cymergee{} operation}
         \label{fig:CycleMerge}
     \end{subfigure}
    \caption{Path merging types.}
    \label{fig:operations}
\end{figure*}
%
%

\paragraph*{Paths Classification. } For the sake of the analysis of the algorithm we classify the paths that the algorithm maintains as \emph{sociable paths} (\`{a} la Dahlhaus et al.~\cite{DHK93}) or as \emph{cycled paths}, as follows. 
\begin{definition}[\cite{DHK93}]\label{def:socint}
    Let $P=(u,\ldots,v)$ be a path in a graph $G$. We say that the path $P$ is \emph{sociable} if
            $d_{P}(u)+d_{P}(v)+1\leq |V(P)|$.
\end{definition}

\begin{definition}[Cycled Path]
Let $P=(u_1,\ldots,u_{\ell})$ be a path in a graph $G$. We say that the path $P$ is \emph{cycled} if $\{u_1, u_{\ell}\} \in E$ or if  there exists an edge of $P$, $\{u_i, u_{i+1}\}$ such that both $\{u_1, u_{i+1}\}, \{u_i, u_{\ell}\} \in E$.
\end{definition}
It is easy to see that if a path $P$ is cycled, then the subgraph induced on $V(P)$ is Hamiltonian.

\noindent
\ifnum\arxiv=1
    We shall use the following basic claim. 
\else
    We shall use the following basic claim, the proof of which is deferred to Appendix~\ref{app:omittedpf}.
\fi
\begin{claim}\label{lemma:basic}
    Let $P=(u_1,\ldots,u_{\ell})$ be a path that is not cycled. Then $P$ is sociable. 
\end{claim}
\ifnum\arxiv=1
    \begin{proof}
    We prove the contra-positive claim.
    Assume the $P$ is not sociable. 
    Then  $d_{P}(u_{\ell})> (|V(P)|-1) - d_{P}(u_1)$.
    
    Consider the $d_{P}(u_1)$ vertices on $P$, $u_j$, that $u_1$ is connected to.
    For each such $u_j$, if $u_{\ell}$ is connected to $u_{j-1}$ then it follows that $P$ is cycled.
    Otherwise, $d_{P}(u_{\ell})$ is at most $(|V(P)|-1) - d_{P}(u_1)$ in contradiction to our assumption.
    \end{proof}
\fi
%
%
%
%
%
%


\section{The Algorithm}
In this section, we list our distributed algorithm (see Algorithm~\ref{alg:hamalg}) without giving all the details of implementation. We then prove its correctness in Section~\ref{sec:correct} and in Appendix~\ref{sec:imp} we discuss in more detail how the algorithm is implemented in the \congest\ model.
We establish the following theorem.
{\renewcommand\footnote[1]{}\mainham*}
\subsection{Listing of the Distributed Algorithm}

Our algorithm begins with finding an initial path-cover of the graph in which each path is of length at least $2$, denoted by $\mathcal{P}_0$. Then the algorithm proceeds in $\Theta(\log n)$ iterations (in which the size of the path-cover decreases by a constant fraction with constant probability). 

We denote the path-cover at the beginning of the $i$-th iteration by $\mathcal{P}_i$.
At the beginning of each iteration, the paths in $\mathcal{P}_i$ are partitioned into pairs. 
Each pair of paths $(P, Q)$ such that $P$ and $Q$ are connected with an edge, and both $P$ and $Q$ are cycled are merged (Step~\ref{step.cyc}).
We denote by $\mathcal{P}^a_i$ the resulting path-cover.

Thereafter, each edge and endpoint of a path in $\mathcal{P}^a_i$ reserves itself to an endpoint of another path in $\mathcal{P}^a_i$, which is picked uniformly at random.

Then, each path $P$ in the path-cover selects out of the elements that were reserved to it (i.e., either an edge or an endpoint) a single element, $\ell$, such that $P$ can be merged to another path via $\ell$ (we refer such elements as {\em useful} for $P$).

\begin{figure}
    \centering
         \begin{subfigure}[b]{0.49\textwidth}
         \centering
        \includegraphics[width=.8\textwidth]{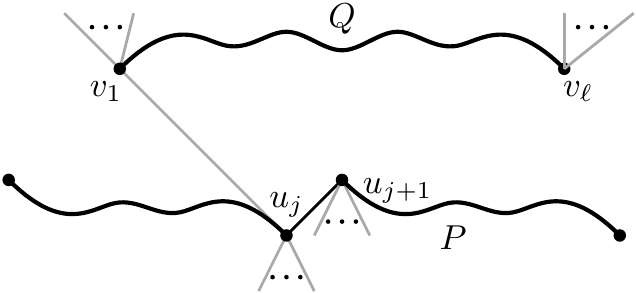}
         \caption{The vertex $u_j$ reserves the edge $\{u_j, u_{j+1}\}$ for $v_1$ (an endpoint of $Q$). However, $\{u_j, u_{j+1}\}$ is non-useful for $Q$ since $v_{\ell}$ (the other endpoint of $Q$) is not connected with an edge to $u_{j+1}$.}
         \label{fig:NotReserve}
     \end{subfigure}
     \hfill
     \begin{subfigure}[b]{0.49\textwidth}
         \centering
        \includegraphics[width=.8\textwidth]{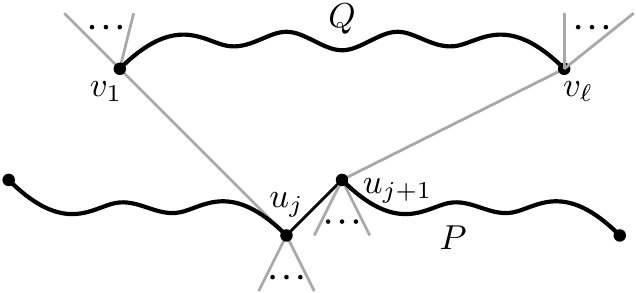}
         \caption{The vertex $u_j$ reserves the edge $\{u_j, u_{j+1}\}$ for $v_1$ (an endpoint of $Q$). The edge $\{u_j, u_{j+1}\}$ is useful for $Q$ since $v_{\ell}$ is connected with an edge to $u_{j+1}$ and so $Q$ can be merged into $P$ via $\{u_j, u_{j+1}\}$.}
         \label{fig:Reserve}
     \end{subfigure}
    \caption{Useful versus non-useful reservations.}
    \label{fig:reserve}
\end{figure}

Each path, $P$, tosses a random fair coin, $y_P$.
Next, all the merges are performed simultaneously where 
a path $P$ is merged via its selected element, $\ell$, to a path $Q$ only if $y_P = \tails{}$ and $y_Q = \heads{}$.  

\begin{remark}
 When a path $P$ is merged to a path $Q$ where $y_P = \tails{}$ and $y_Q = \heads{}$, the orientation of $Q$ is kept, and the orientation of $P$ may be reversed to maintain consistency with the orientation of $Q$. For example, if $P = (y_1, \ldots, y_\ell)$ is merged to $Q$ via $\{u, v\}$ where the orientation of $\{u, v\}$ is from $u$ to $v$ and $y_1$ is connected to $v$ and $y_\ell$ is connected to $u$ then the orientation of $P$ is reversed after the merge. 
 \end{remark}
 
\begin{algorithm}
\DontPrintSemicolon
  \KwInput{A Dirac graph $G=(V, E)$.}
  \KwInit{The algorithm maintains a path-cover $\mathcal{P}_i$ for every $i\geq 0$. 
  The path-cover $\mathcal{P}_{i+1}$ is computed from $\mathcal{P}_i$ via \elmerge{}, \cymerge{} and \conmerge{} operations. }
  \KwOutput{$\mathcal{P}_{\Theta(\log n)}$ is a Hamiltonian path. w.h.p.\tcp*[l]{For the exact constant within the $\Theta$ notation we refer the reader to Lemma~\ref{lemma:numofiter}.}}
\medskip
Compute a path-cover, $\mathcal{P}_0$ in the graph such that each path is of size at least $2$.\label{step1}\;
\For {$i \gets 0$ \KwTo $\ell = \Theta(\log n)$}{\label{step.for}
    %
    Pair all paths (except for at most one) and let $\mathcal{I}_i$ denote the set of these pairs of paths.\label{step:pair}\;
    Let $\mathcal{L}_i\subseteq \mathcal{I}_i$ denote the set of paired paths that have an edge between them and for which both paths are cycled.\;
    %
    %
    %
    Perform \cymerge{} on all pairs in $\mathcal{L}_i$.  \label{step.cyc}\tcp*[l]{see Definition~\ref{def:cyclemerge}}
    %
    %
    Let $\mathcal{P}^a_i$ denote the current path-cover\;
For every $v\in V$ let $D_i(v)$ denote the subset of endpoints of paths in $\mathcal{P}^a_i$, $u$, such that $\{u, v\}\in E$\;

\For{every $P\in \mathcal{P}^a_i$ where $P=(u_1, \ldots, u_\ell)$\label{step:pi}}{
Every $u_j$ for $j\in \{1,\dots, \ell-1\}$ picks an endpoint $v$ u.a.r. from $D_i(u_j)$ and \emph{reserves} the edge $(u_j, u_{j+1})$ for $v$.\;

Additionally, each endpoint of $P$, $v$, reserves itself to an endpoint which is picked u.a.r. from $D_i(v)$. \label{step:di}\;

$y_P \gets \begin{cases}\heads{}, & \text{w.p. } 1/2, \\ \tails{}, & \text{o.w.} \end{cases}$.\label{step.toss}\;
}

$M_i \gets \emptyset$\label{step.init} \;

\For{every $P\in \mathcal{P}^a_i$}{
We say that an endpoint or an edge, $x$, is useful for $P$ (w.r.t. $\mathcal{P}^a_i$) if $P$ can be concatenated or merged along $x$ to another path in $\mathcal{P}^a_i$.

Let $S_i(P)$ denote the set of elements reserved for the endpoints of $P$ in Steps~\ref{step:pi}-\ref{step:di} which are also {\em useful} for $P$.\;

$P$ picks arbitrarily one of the elements  in $S_i(P)$ (assuming it is not empty) and adds it to $M_i$.\label{step:pick}\;
    }
    Perform \conmerge{} and \elmerge{} with accordance to $M_i$ and the $y_P$ variables: a path $P_1$ merges into a path $P_2$ if there is a corresponding merge operation in $M_i$ and if $y_{P_1}=\tails{}$ and $y_{P_2}=\heads{}$. \label{step.merge1} \tcp*[l]{see Definitions~\ref{def:elem},~\ref{def:concat}}\label{step.orientmerge}
    %
    %
     %
    %
}
    \Return $\mathcal{P}_{\ell}$.\label{step.return}
\caption{Finding a Hamiltonian path in a Dirac graph.}\label{alg:hamalg}
\end{algorithm}
\section{Correctness of the Algorithm}\label{sec:correct}

In this section, we prove the correctness of our algorithm. 

We begin by giving a lower bound on the number of possible merges for each path in the path-cover. 
We then use this bound to give a lower bound on the expected number of merges carried out in each iteration. 
Finally, we prove that with high probability after $\Theta(\log n)$ iterations, all paths are merged into a single path.
\ifnum\arxiv=0
Some of the proofs are deferred to Appendix~\ref{app:omittedpf}.
\fi
\subsection{Number of Possible Merges for Good Paths}

In this section, we provide the proof of Lemma~\ref{lemma.paths} which gives a lower bound on the number of possible ways a path $P$ can be merged into a path $Q$. 
We then define the notion of {\em good} paths (Definition~\ref{def:good}). 
Roughly speaking, a path is good if its endpoints are neighbors of sufficiently many vertices that belong to other paths in the path cover. 
Thereafter, we use Lemma~\ref{lemma.paths} to give a lower bound on the total number of ways a good path, $P$, can be merged into any other path in the path-cover (Lemma~\ref{lemma:leftdeg}). 

Let $P=(u,\ldots,v)$ and $Q=(a,\ldots, b)$ be disjoint paths in $G$. %
Let $M(P, Q)$ denote the number of edges along which one can merge $P$ into $Q$ via elementary merging plus the number of vertices along which one can concatenate $P$ to $Q$ (see Definitions~\ref{def:elem} and~\ref{def:concat}, respectively).  %
Let $d(P, Q)$ denote the number of endpoints of $Q$ connected with an edge to an endpoint of $P$.
%


We begin with extending the following lemma from~\cite{DHK93}. 
\begin{lemma}[Lemma.~5.2.5~\cite{DHK93}, restated]\label{lemma:lemma525}
    Let $P=(u,\ldots,v)$ and $Q$ be disjoint paths in $G$.
    If  $d(P,Q)=0$, then the number of edges along which one can merge $P$ into $Q$ via elementary merging operations is at least $d_{Q}(u)+d_{Q}(v)-|V(Q)|+1$.
\end{lemma}

To achieve better round complexity, our algorithm performs concatenation and elementary merges simultaneously. To this end, we prove the following lemma, which extends Lemma~\ref{lemma:lemma525} so that it also applies to cases where $d(P, Q)>0$. 
\begin{lemma}\label{lemma.paths}
    Let $P=(u,\ldots,v)$ and $Q=(a,\ldots, b)$ be disjoint paths in $G$ such that $|V(P)|, |V(Q)| \geq 2$.  Then,
    \begin{align}
        M(P,Q) \geq d_{Q}(u)+d_{Q}(v)-|V(Q)|+1 \:.
    \end{align}
\end{lemma}
\ifnum\arxiv=1
\begin{proof}
    We consider two cases: (1)~$d(P,Q)=0$, and (2)~$d(P,Q)>0$.

    The case where $d(P, Q)=0$ follows from Lemma~\ref{lemma:lemma525}.

    Let us assume that  $d(P,Q)=1$, and that w.l.o.g. one of the endpoints of $P$ is adjacent to $a$. We add a new vertex $a'$ and a new edge $\{a',a\}$ to $Q$. Let $Q'$ denote the resulting path. Now, paths $P$ and $Q'$ satisfy the conditions of Lemma~\ref{lemma:lemma525}, i.e., $d(P,Q')=0$. It holds that $M(P,Q') = M(P,Q) - 1$ since by extending the path $Q$ we eliminate a single possible concatenation operation while not changing the possible merging operations. Now, due to the first case of this proof and since $d_{Q'}(u)=d_{Q}(u)$ and $d_{Q'}(v)=d_{Q}(v)$ we obtain that 
    \begin{align*}
        d_{Q'}(u)+d_{Q'}(v)-|V(Q)\cup \{a'\}|+1 &\leq M(P,Q') \Leftrightarrow \\
        d_{Q}(u)+d_{Q}(v)-|V(Q)|-1+1 &\leq M(P,Q)-1\:, 
    \end{align*} 
    as required.
    For the case where $d(P,Q)=2$ we simply apply the previous case twice. The lemma follows. 
\end{proof}
\fi

We next give a lower bound on the number of merging operations which applies for a subset of paths in the path cover. 
We shall use the following definition. 

\begin{definition}
   Let $\mathcal{P}$ be a path cover and let $A(\mathcal{P})$ denote the set of paths, $P=(u,\ldots, v)$, in  $\mathcal{P}$ such that there exists a path $P'\in \mathcal{P}$ for which $|V(P')| \geq |V(P)|$ and $d_{P'}(u) + d_{P'}(v) = 0$.
\end{definition}
In words, $A(\mathcal{P})$ is the set of paths, $P \in\mathcal{P}$ for which there exists another path, $P'\in \mathcal{P}$, which is not shorter than $P$ and for which the endpoints of $P$ are not adjacent to any of the vertices composing $P'$. 
We next define the notion of {\em good path}.

\begin{definition}[good path]\label{def:good}
   Let $\mathcal{P}$ be a path cover. A path $P\in \mathcal{P}$ is {\em good} (w.r.t. $\mathcal{P}$) if it is sociable or in $A(\mathcal{P})$.
\end{definition}

The following lemma gives a lower bound on the number of merging operations for any path which is good.
In the proof of the lemma, we extend ideas from Corollary~5.2.6 and  Lemma~5.2.7 of Dahlhaus et al.~\cite{DHK93}. Our extension allows us to drop the stringent requirement of Dahlhaus et al.~\cite{DHK93} that the endpoints of $P$ are not adjacent to endpoints of any other path in the path cover. This allows us to support the execution of both concatenations and elementary merging operations simultaneously.  

\begin{lemma}\label{lemma:leftdeg}
Let $\mathcal{P}$ be a path-cover of $G$. 
For $P\in \mathcal{P}$ define $M(P) \eqdef \sum_{Q\in \mathcal{P}\setminus \{P\}} M(P, Q)$. 
    Then, for every good path, $P\in \mathcal{P}$ it holds that $M(P) \geq |\mathcal{P}|$.
\end{lemma}
\ifnum\arxiv=1
\begin{proof}
Let $P=(u,\ldots,v)$.
We first consider the case in which $P$ is sociable. Namely, the case in which $d_P(u) + d_P(v) -|V(P)| +1 \leq 0$.
By Lemma~\ref{lemma.paths}, 
\begin{align*}
M(P) \geq& 
\sum_{Q\in \mathcal{P}\setminus\{P\}} \left( d_Q(u) + d_Q(v) -|V(Q)| +1 \right) 
\\
=&\left(\sum_{Q\in \mathcal{P}} \left(d_Q(u) + d_Q(v) -|V(Q)| +1 \right) \right) \\
&- \left( d_P(u) + d_P(v) -|V(P)| +1\right).
\end{align*}
\sloppy Since $\sum_{Q\in \mathcal{P}} \left(d_Q(u) + d_Q(v)\right) = d(u) + d(v) \geq n$ and $\sum_{Q\in \mathcal{P}} |V(Q)| = n$ it follows that 
\begin{align}\label{eq.P}
\sum_{Q\in \mathcal{P}} \left(d_Q(u) + d_Q(v) -|V(Q)| +1 \right) \geq |\mathcal{P}|.
\end{align}
Thus, $M(P) \geq  |\mathcal{P}|- \left( d_P(u) + d_P(v) -|V(P)| +1\right)$. 
Therefore, the claim follows from the fact that $P$ is sociable. 

We now consider $P\in A(\mathcal{P})$. 
Let $P'$ denote the path in $\mathcal{P}$ such that $|V(P')| \geq |V(P)|$ and $d_{P'}(u) + d_{P'}(v) = 0$.
\begin{align*}
\sum_{Q\in \mathcal{P}\setminus \{P, P'\}} M(P, Q) \geq& 
\sum_{Q\in \mathcal{P}\setminus\{P, P'\}} \left( d_Q(u) + d_Q(v) -|V(Q)| +1 \right) 
\\
=&\left(\sum_{Q\in \mathcal{P}} \left(d_Q(u) + d_Q(v) -|V(Q)| +1 \right) \right) \\
&
- \left( d_P(u) + d_P(v) -|V(P)| +1\right)\\
&- \left( d_{P'}(u) + d_{P'}(v) -|V(P')| +1\right).
\end{align*}
Since $d_{P'}(u) + d_{P'}(v) = 0$ and $d_P(u) + d_P(v) \leq 2(|V(P)| -1)$ it follows that 
$\left( d_P(u) + d_P(v) -|V(P)| +1\right)
+ \left( d_{P'}(u) + d_{P'}(v) -|V(P')| +1\right)\leq 0$.

Thus by Equation~\eqref{eq.P}, 
$M(P) \geq \sum_{Q\in \mathcal{P}\setminus \{P, P'\}} M(P, Q) \geq |\mathcal{P}|
$, 
as desired. This concludes the proof of the lemma.
\end{proof}
\fi


\subsection{Expected Number of Merges}
In this section, we prove Claim~\ref{claim:expapxmatch} which states that the expected number of merges of each iteration is sufficiently large.

More specifically, for a fixed iteration $i$, we prove that if the number of cycle-merges performed at Step~\ref{step.cyc} is below some threshold (specifically if  $|\mathcal{L}_i| \leq |\mathcal{P}_i|/12$), then the expected size of $M_i$ (the set of concatenation merges and elementary merges added in Step~\ref{step:pick}) is a constant fraction of the size of $\mathcal{P}_i$. We first prove the following claim which gives a lower bound on the probability that a good path receives a useful reservation. 

\begin{claim}\label{clm:single}
Fix an iteration $i$. For any $P$ which is good with respect to $A(\mathcal{P}^a_i)$, it holds that $P$ receives a useful reservation for at least one element with probability at least $1/3$. 
\end{claim}
\begin{proof}
Fix an iteration $i$ and let $P$ be a good path with respect to $A(\mathcal{P}^a_i)$. 
Let $F$ denote the set of edges and endpoints that $P$ can be merged to (see Definitions~\ref{def:elem}, \ref{def:concat}) in $\mathcal{P}^a_i$.
By Lemma~\ref{lemma:leftdeg}, it holds that $|F| \geq |\mathcal{P}^a_i|$.
Let $x=2|\mathcal{P}^a_i|$ denote the total number of endpoints of paths in $\mathcal{P}^a_i$.
Since every edge and endpoint in $F$ is reserved for $P$ with a probability of at least $1/x$, the probability that $P$ received at least one reservation of an element in $F$ is at least $1 -(1-1/x)^{|F|}$.
since $(1-1/x)^{|F|} \leq (1-1/(2|\mathcal{P}^a_i|))^{|\mathcal{P}^a_i|} \leq 1/\sqrt{e}$, it holds that this probability is at least $1/3$.
\end{proof}
\begin{claim}\label{claim:expapxmatch}
If $\mathcal{L}_i \leq |\mathcal{P}_i|/12$ then $\Exp(|M_i| ~\mid~ \mathcal{P}_i) \geq  |\mathcal{P}_i|/36$.
\end{claim}
\begin{proof}
We first observe that in every pair $(P, Q) \in \mathcal{I}_i\setminus \mathcal{L}_i$ at least one path is good with respect to $\mathcal{P}^a_i$. To see this, observe that there are two cases. The first case is that $P$ and $Q$ are not connected with an edge. This implies that at least one of them is in $A(\mathcal{P}^a_i)$ The second case is that at least one of them is not cycled, which by Claim~\ref{lemma:basic} implies that at least one of them is sociable, as desired.                                       
We next lower bound the number of pairs in $\mathcal{I}_i\setminus \mathcal{L}_i$. Since the number of pairs is at least $(|\mathcal{P}_i|-1)/2 \geq |\mathcal{P}_i|/4$ and $|\mathcal{L}_i| < |\mathcal{P}_i|/12$ it holds that $|\mathcal{I}_i\setminus \mathcal{L}_i| \geq |\mathcal{P}_i|/4 - |\mathcal{P}_i|/12 = |\mathcal{P}_i|/6$.

Therefore, by Claim~\ref{clm:single}, the expected size of merges that are added to $M_i$ is at least $|\mathcal{I}_i\setminus \mathcal{L}_i|/2 \cdot (1/3) \geq |\mathcal{P}_i|/36$, as required.
\end{proof}

\subsection{The Progress of Each Iteration}\label{sec:progress}

In this section, we prove the following lemma.
\begin{lemma}\label{lemma:iter}
    For any iteration $i$ of   Algorithm~\ref{alg:hamalg} it holds that   $\Exp(|\mathcal{P}_{i+1}|) \leq (1-1/144)\cdot \Exp(|\mathcal{P}_i|)$.
\end{lemma}

\begin{proof}
Fix an iteration $i$. If $\mathcal{L}_i \geq |\mathcal{P}_i|/12$ then at least $|\mathcal{P}_i|/12$ paths of $\mathcal{P}_i$ are merged and so $|\mathcal{P}_{i+1}| \leq (1-1/12)|\mathcal{P}_i|$, as desired.

Otherwise, by Claim~\ref{claim:expapxmatch}, $\Exp(|M_i|) \geq |\mathcal{P}_i|/36$.
Consider a merge operation in $M_i$ in which path $P$ is merged into path $Q$. This merge is carried in Step~\ref{step.orientmerge} only if $y_P = \tails{}$ and $y_Q = \heads{}$, which happens with probability $1/4$. We denote these merge operations by $M'_i$, hence $\Exp(|M'_i|)|\geq \mathcal{P}_i|/(36\cdot 4) =  |\mathcal{P}_i|/144$. Moreover, since $|\mathcal{P}_{i+1}| \leq |\mathcal{P}_i|-|M'_i|$ (recall that merges can occur before Step~\ref{step.orientmerge}), it follows that $\Exp(|\mathcal{P}_{i+1}|~\mid~\mathcal{P}_{i}) \leq |\mathcal{P}_i|-\Exp(|M'_i|~\mid~\mathcal{P}_{i})\leq (1-1/144)\cdot |\mathcal{P}_{i}|$.

The lemma follows since
\[
    \Exp(|\mathcal{P}_{i+1}|)=\sum_{\mathcal{P}_{i}}\Pr(\mathcal{P}_{i})\cdot\Exp(|M'_i|~\mid~\mathcal{P}_{i})\leq \sum_{\mathcal{P}_{i}}\Pr(\mathcal{P}_{i})\cdot(1-1/144)\cdot |\mathcal{P}_{i}|=(1-1/144)\cdot\Exp(|\mathcal{P}_{i}|)\:,
\] as required. 
\end{proof}

\subsubsection{Number of Iterations}
In this section, we prove that if the for-loop in Step~\ref{step.for} of Algorithm~\ref{alg:hamalg} performs $\Theta(\log n)$ iterations, then the path that is returned at the end of the algorithm is Hamiltonian w.h.p. 

We note that, although the algorithm uses independent coin tosses between different  iterations, the success of two different iterations are random variables that may be dependent. Therefore we cannot use concentration bounds that assume the independence of the random variables (such as Chernoff's bound).
Roughly speaking, the dependence comes from the fact that the constructed path cover depends on the coin tosses of previous iterations.
Nonetheless, we can show that $\Theta(\log n)$ iterations are sufficient. 
The proof of our main theorem (Theorem~\ref{thm:ham}) follows directly from the following lemma. 

%
\begin{lemma}\label{lemma:numofiter}
    The path returned in Step~\ref{step.return} is Hamiltonian w.h.p. 
\end{lemma}
\begin{proof}
    Lemma~\ref{lemma:iter} and the fact that $|\mathcal{P}_0| \leq n$ imply that 
    \[
        \Exp(|\mathcal{P}_{\ell}|) \leq (1-1/144)^{\ell}\cdot \Exp(|\mathcal{P}_0|)\leq (1-1/144)^{\ell}\cdot n\:.
    \]
    It follows that for $\ell \geq \frac{2}{\log_2 (144/143)}\cdot \log_2 n$, it holds that $\Exp(|\mathcal{P}_{\ell}|) \leq \frac{1}{n}$. 

    Thus, by Markov's Inequality, it follows that the probability that the Algorithm fails is 
    \[
        \Pr(|\mathcal{P}_{\ell}|>1) \leq \frac{\Exp(|\mathcal{P}_{\ell}|)}{1} \leq \frac{1}{n}\:,
    \] as required.
\end{proof}
%
%
%
%
%
\ifnum\arxiv=1
\paragraph*{Martingales and Azuma-Hoeffding Inequality~\cite{MU05}.} \sloppy A sequence of random variables $Z_0, Z_1, \ldots$ is called \emph{martingale} if for all $n\geq 0$ it holds: $\Exp{(|Z_n|)}<\infty$, and $\Exp{(Z_{n+1} \mid Z_0, \ldots, Z_n)}=Z_n$. 
The following theorem states a simplified version of the Azuma-Hoeffding Inequality. 
\begin{theorem}[Coro.~12.5~\cite{MU05}, restated]\label{thm:azuma}
    Let $Z_0, Z_1, \ldots$ a martingale. If there is $c\geq 0$ for all $k\geq 1$ such that
    $|Z_k - Z_{k-1}| \leq c$.  
    Then, for all $t\geq 1$ and $\lambda \geq 0$,
    $
        \Pr{\left(Z_t-Z_0\geq \lambda c \sqrt{t}\right)}\leq e^{-\lambda^2/2}
    $. 
\end{theorem}

\paragraph*{Application of the Azuma-Hoeffding Inequality. }
In the following we bound the number of successful iterations (from above): we first define a sequence of random variables $Z_0, Z_1, \ldots$, and show that it is a martingale. We then apply the Azuma-Hoeffding Inequality and show that in $\Theta(\log n)$  iterations there are $\Theta(\log n)$ successful iterations, that is, Algorithm~\ref{alg:hamalg} is successful and returns a Hamiltonian path. 

For $i\geq 1$, let $Y_i$ be the random variable that takes the value $-1/p_i +1$ if the $i$-th iteration is successful and takes the value $1$ otherwise (note that $Y_i$ can take negative values). Let $Y_0 \triangleq  1$, and let $Z_i \triangleq \sum_{j=0}^i Y_j$, then 
$\E(Z_{k} \mid Z_0, \ldots Z_k) =  Z_{k-1} + 1\cdot (1-p_k) + (1-1/p_k)\cdot p_k = Z_{k-1}$,  
and $|Z_k - Z_{k-1}| \leq \frac{1}{\eta}$ for every $k\geq 1$. Hence, the sequence $Z_0, Z_1, \ldots$ is indeed a martingale. Applying Theorem~\ref{thm:azuma} with
$\lambda = \frac{\eta}{2}\sqrt{t}$, $c = \frac{1}{\eta}$, and $t = \frac{8}{\eta^2}\ln n$
yields $\Pr\left(Z_{t}-1 \geq \frac{\eta}{2}\sqrt{t}\cdot\frac{1}{\eta}\cdot\sqrt{t}\right) \leq \frac{1}{n}$. 
Hence,  
$\Pr{\left(Z_{t} - 1 < \frac{t}{2}\right)} \geq 1-\frac{1}{n}$. 
Let $s$ denote the number of successful iterations. Since for every $ i \geq 1$ it holds that $\frac{1}{p_i} \leq \frac{1}{\eta}$, then $Z_{t} - 1 \geq \left(t - s\right) + s \left(1-\frac{1}{\eta}\right) = t - \frac{s}{\eta}$.  
Thus, we obtain that $t - \frac{s}{\eta} < \frac{t}{2}$ w.p. at least $1-\frac{1}{n}$.
We conclude that with probability at least $1-\frac{1}{n}$,
 the number of successful iterations is at least
 $
    t\cdot \eta - \frac{t\cdot\eta}{2} = \frac{4}{\eta}\cdot \ln n \geq 4\log_{1/(1-\eta)} n$, 
where the last inequality follows since 
$\ln(x) \geq 1-\frac{1}{x}$, and since
    $\log_{1/(1-\eta)}n  = \frac{\ln n}{\ln \frac{1}{1-\eta}}$, 
    $\ln\left(\frac{1}{1-\eta}\right)  
     \geq 1 - \frac{1}{\frac{1}{1-\eta}} = \eta$, 
as required.
\fi 

%
%

\ifnum\arxiv=1
    \bibliographystyle{alpha}
\fi
\bibliography{hamcong.bib}
\appendix
\section{Other Related Work}\label{app:related}
\paragraph*{Parallel Algorithms. }
As mentioned above, Dahlhaus et al.~\cite{DHK93} gave a $O(\log^4n)$ \cpram\ algorithm that uses a linear number of processors to find a Hamiltonian cycle in Dirac graphs. Another generalization of Dirac graphs are \emph{Chv\'{a}tal graphs}. A graph is called a Chv\'{a}tal graph if its degree sequence $d_1 \leq d_2 \leq \cdots \leq d_n$ satisfies that for every $k < n/2$, $d_k \leq k$ implies that $d_{n-k} \geq n-k$. Chv\'{a}tal proved (see e.g.,~\cite{bollobas2004extremal}) that Chv\'{a}tal graphs are also Hamiltonian. In~\cite{bondy1976method} a sequential polynomial time algorithm that finds Hamiltonian cycles in Chv\'{a}tal graphs has been shown. S\'{a}rk\"{o}zy~\cite{sarkozy2009fast} proved that a deterministic $O(\log^4 n)$ time \epram\  algorithm with a polynomial number of processors that finds a Hamiltonian cycle in a \emph{$\eta$-Chv\'{a}tal graphs} exists. A graph is called $\eta$-Chv\'{a}tal graphs if for every $k < n/2$, $d_k \leq \min\{k+\eta n,n/2\}$, it holds that $d_{n-k-\eta n} \geq n-k$, where $0 < \eta < 1$.
\section{Deferred Proofs}\label{app:omittedpf}
\subsection{Proof of Lemma~\ref{lemma:basic}}
\begin{proof}
We prove the contra-positive claim.
Assume the $P$ is not sociable. 
Then  $d_{P}(u_{\ell})> (|V(P)|-1) - d_{P}(u_1)$.

Consider the $d_{P}(u_1)$ vertices on $P$, $u_j$, that $u_1$ is connected to.
For each such $u_j$, if $u_{\ell}$ is connected to $u_{j-1}$ then it follows that $P$ is cycled.
Otherwise, $d_{P}(u_{\ell})$ is at most $(|V(P)|-1) - d_{P}(u_1)$ in contradiction to our assumption.
\end{proof}
\subsection{Proof of Lemma~\ref{lemma.paths}}
\begin{proof}
    We consider two cases: (1)~$d(P,Q)=0$, and (2)~$d(P,Q)>0$.

    The case where $d(P, Q)=0$ follows from Lemma~\ref{lemma:lemma525}.

    Let us assume that  $d(P,Q)=1$, and that w.l.o.g. one of the endpoints of $P$ is adjacent to $a$. We add a new vertex $a'$ and a new edge $\{a',a\}$ to $Q$. Let $Q'$ denote the resulting path. Now, paths $P$ and $Q'$ satisfy the conditions of Lemma~\ref{lemma:lemma525}, i.e., $d(P,Q')=0$. It holds that $M(P,Q') = M(P,Q) - 1$ since by extending the path $Q$ we eliminate a single possible concatenation operation while not changing the possible merging operations. Now, due to the first case of this proof and since $d_{Q'}(u)=d_{Q}(u)$ and $d_{Q'}(v)=d_{Q}(v)$ we obtain that 
    \begin{align*}
        d_{Q'}(u)+d_{Q'}(v)-|V(Q)\cup \{a'\}|+1 &\leq M(P,Q') \Leftrightarrow \\
        d_{Q}(u)+d_{Q}(v)-|V(Q)|-1+1 &\leq M(P,Q)-1\:, 
    \end{align*} 
    as required.
    For the case where $d(P,Q)=2$ we simply apply the previous case twice. The lemma follows. 
\end{proof}
\subsection{Proof of Lemma~\ref{lemma:leftdeg}}
\begin{proof}
Let $P=(u,\ldots,v)$.
We first consider the case in which $P$ is sociable. Namely, the case in which $d_P(u) + d_P(v) -|V(P)| +1 \leq 0$.
By Lemma~\ref{lemma.paths}, 
\begin{align*}
M(P) \geq& 
\sum_{Q\in \mathcal{P}\setminus\{P\}} \left( d_Q(u) + d_Q(v) -|V(Q)| +1 \right) 
\\
=&\left(\sum_{Q\in \mathcal{P}} \left(d_Q(u) + d_Q(v) -|V(Q)| +1 \right) \right) \\
&- \left( d_P(u) + d_P(v) -|V(P)| +1\right).
\end{align*}
\sloppy Since $\sum_{Q\in \mathcal{P}} \left(d_Q(u) + d_Q(v)\right) = d(u) + d(v) \geq n$ and $\sum_{Q\in \mathcal{P}} |V(Q)| = n$ it follows that 
\begin{align}\label{eq.P}
\sum_{Q\in \mathcal{P}} \left(d_Q(u) + d_Q(v) -|V(Q)| +1 \right) \geq |\mathcal{P}|.
\end{align}
Thus, $M(P) \geq  |\mathcal{P}|- \left( d_P(u) + d_P(v) -|V(P)| +1\right)$. 
Therefore, the claim follows from the fact that $P$ is sociable. 

We now consider $P\in A(\mathcal{P})$. 
Let $P'$ denote the path in $\mathcal{P}$ such that $|V(P')| \geq |V(P)|$ and $d_{P'}(u) + d_{P'}(v) = 0$.
\begin{align*}
\sum_{Q\in \mathcal{P}\setminus \{P, P'\}} M(P, Q) \geq& 
\sum_{Q\in \mathcal{P}\setminus\{P, P'\}} \left( d_Q(u) + d_Q(v) -|V(Q)| +1 \right) 
\\
=&\left(\sum_{Q\in \mathcal{P}} \left(d_Q(u) + d_Q(v) -|V(Q)| +1 \right) \right) \\
&
- \left( d_P(u) + d_P(v) -|V(P)| +1\right)\\
&- \left( d_{P'}(u) + d_{P'}(v) -|V(P')| +1\right).
\end{align*}
Since $d_{P'}(u) + d_{P'}(v) = 0$ and $d_P(u) + d_P(v) \leq 2(|V(P)| -1)$ it follows that 
$\left( d_P(u) + d_P(v) -|V(P)| +1\right)
+ \left( d_{P'}(u) + d_{P'}(v) -|V(P')| +1\right)\leq 0$.

Thus by Equation~\eqref{eq.P}, 
$M(P) \geq \sum_{Q\in \mathcal{P}\setminus \{P, P'\}} M(P, Q) \geq |\mathcal{P}|
$, as desired. This concludes the proof of the lemma.
\end{proof}
\section{A Detailed Description of the Distributed Implementation of Algorithm~\ref{alg:hamalg}}\label{sec:imp}

In this section, we elaborate more on how each algorithm step is implemented in the distributed \congest\ model. We describe the construction and maintenance of spanning trees, which our algorithm implicitly uses.

\subsection{Spanning Trees}\label{sec:spanningtrees}
Algorithm~\ref{alg:hamalg} implicitly uses spanning trees for coordination of the paths and coordination of vertices within a path. In this section, we describe the spanning trees the algorithm uses and how it constructs and maintains them. We begin with the following basic claim.

\begin{claim}\label{clm.clm15}
Let $G=(V, E)$ be a graph. Then for any pair of vertices $u$, $v$ such that $d(u)+ d(v) \geq n$ it holds that $\delta(u, v) \leq 2$. 
\end{claim}

\begin{proof}
Let $G=(V, E)$ be a graph, and let $v\in V$ and $u\in V$ be a pair of vertices in $G$. 
If $u\in N(v)$ then the claim follows. Otherwise, since $d(u)+ d(v) \geq n$ and $u\notin N(u)\cup N(v)$ it has to be that $N(u)\cap N(v) \neq \emptyset$. The claim follows. 
\end{proof}

\begin{corollary}\label{clm:dia}
The diameter of a Dirac's graph is at most $2$.
\end{corollary}

\subsubsection{The Global Spanning Tree}\label{sec:globalspan}
The first tree the algorithm uses is a tree that spans the entire graph. This tree remains fixed throughout the entire execution of the algorithm. 
We first pick the leader of the graph to be the vertex whose identifier is minimal. This takes $O(D)$ rounds, where $D$ denotes the diameter of the graph (which is a constant in our case). Once the leader has been selected, a spanning tree of the graph is constructed in $O(D)$ rounds. By Corollary~\ref{clm:dia} the spanning tree has depth at most $2$. 

The global spanning tree is used implicitly for pairing the paths in Step~\ref{step:pair} (see Section~\ref{sec:pardist} for more details).

\subsubsection{Maintaining a Spanning Tree for Each Path in $\mathcal{P}_i$}

Algorithm~\ref{alg:hamalg} also implicitly uses  spanning trees of depth at most $2$ for spanning each one of the paths in the path-cover.

Specifically, each iteration $i$ starts such that each path $P\in \mathcal{P}_i$ has a spanning tree (of depth at most $2$). We denote the spanning tree corresponding to $P$ by $T(P)$. 
The root of $T(P)$ is the leader of $P$ which is defined by the following process.

Initially, the length of each path in ${\mathcal P}_0$ is constant and so the leader of each path can be selected to be the identifier of an arbitrary vertex in the path.

When path $P_1$ is merged into path $P_2$ the new leader of the new (merged) path $P'$ is taken to be the leader of $P_2$. 
By using $T(P_1)$ all vertices in $P_1$ are notified with the id of the new leader.  

The spanning-tree $T(P')$ is then constructed as follows. Let $r$ denote the leader of $P'$.
Each vertex $v$ in $P'$ which is not a neighbor of $r$ sends the id of $r$ to all its neighbors. In response, each $u\in N(v)$ sends to $v$ whether $u$ is a neighbor of $r$ or not. Then $v$ picks arbitrarily $w\in N(v) \cap N(r)$ to be its parent and notifies $w$ about it.

Consider any $P$ and any $e$ which is an edge of $T(P)$. Observe that $e$ has at least one of its endpoints in $P$ (either it is an edge incident to the root or it is an edge which is incident to the leaves of the tree). Therefore each edge can belong to at most $2$ different spanning trees (of paths of the current path-cover). 

Hence, all merges in Step~\ref{step.cyc} as well as for Step~\ref{step.merge1} can be carried out simultaneously without creating congestion.

\subsection{Finding the Initial path-cover (Step~\ref{step1})}
In Step~\ref{step1}, the algorithm computes an initial path-cover $\mathcal{P}_{0}$ such that each path consists of at least two vertices. This path-cover is constructed as follows. 
\paragraph*{First Step.}
At the first step, we find a maximal matching of the graph using the algorithm of Barenboim et al.~\cite{BEPS16}.
Let $F$ denote the set of edges of the matching of this first step.
Let $F'$ denote the set of vertices that were matched in the first step.

\paragraph*{Second Step.}
In the next step, we remove all edges in $F'\times F'$ and find again a maximal matching on the resulting graph.
Let $S$ denote the set of edges of this second matching. 
Let $S'$ denote the set of vertices that were not matched in the first step but were matched in the second step.

\begin{claim}
The subgraph induced by the edges in $F\cup S$ is a path-cover of $G$ in which each path consists of at least $2$ and at most $4$ vertices.
\end{claim}

\begin{proof}
We first prove that every vertex in the graph is in $F'\cup S'$ (which implies that the length of each path in the path-cover is at least $2$).
Assume towards contradiction that there exists $v\in V\setminus (F'\cup S')$.
Since $v$ was not matched in the first step we conclude that all its neighbors are in $F'$ (by the maximality of the matching). 
Thus $|F'| \geq n/2$ (since $d(v) \geq n/2$). 
Since $v$ was not matched in the second step we conclude that all its neighbors were matched again in the second step. Each one of these neighbors had to be matched in the second step to a vertex not in $F'$ (since we removed the edges in $F'\times F'$), thus, we conclude that $|S'| \geq n/2$. Thus $|F'\cup S'| \geq n$, a contradiction. 

To prove an upper bound on the length of each path we first observe that every vertex $v\in S'$ has to be an endpoint of a path (since it belongs only to a single edge in $F\cup S$). Thus the inner vertices of paths are only from $F'$. Since each vertex in $F'$ is matched only to a single vertex in $F'$ (when considering both the first and the second matching) it holds that any path can contain at most $4$ vertices, as desired.  
\end{proof}

\subsection{Pairing the  Paths (Step~\ref{step:pair})}\label{sec:pardist}
The paths pairing is performed using the global spanning tree (see Section~\ref{sec:globalspan}).
More specifically, each leader of a path sends up the tree its identifier. The vertices in the middle layer of the tree pair up all the paths from which they received a message except for at most one (depending on the parity). The leftovers (namely, unpaired paths) are sent up to the root, which pairs them and sends them back to the leaves.

\subsection{Selecting the Merge Operations  (Steps~\ref{step:pi}-\ref{step:pick})}\label{sec:maxm}
We assume that during the entire execution of the algorithm, every vertex $v$ knows to which path it belongs in the current path-cover, and for every $u\in N(v)$, $v$ knows to which path $u$ belongs and whether $u$ is an endpoint of a path.

This can be implemented by letting each vertex send updates to its neighbors after performing each one of the merging steps (i.e., Steps~\ref{step.cyc} and~\ref{step.orientmerge}).
Therefore, we may assume that every $v\in V$ knows $D_i(v)$ at Step~\ref{step:di}.

We next describe how the merging operations are selected. 
For every $P\in \mathcal{P}_i$ where $P=(u_1, \ldots, u_\ell)$ every $u_j$ for $j\in \{1,\dots, \ell-1\}$ picks u.a.r. an endpoint $v \in D_i(u_j)$ and \emph{reserves} the edge $(u_j, u_{j+1})$ for $v$.
Let $Q$ denote the path of $v$. The vertex $u_j$ sends the identity of $Q$ and $v$ to $u_{j+1}$ and if the other endpoint of $Q$, $w$ is in $D_i(u_{j+1})$, then $u_{j+1}$ notifies $w$ that $(u_i, u_{j+1})$ has been {\em reserved} for $Q$. 

Similarly, each endpoint of $P$, $v$, picks u.a.r. a vertex $w\in D_i(v)$ and sends $w$ its identity and a notification of a \emph{reservation}. 

Therefore each edge and an endpoint of a path in $\mathcal{P}_i$ is reserved for at most one path in $\mathcal{P}_i$.

Note that it might be the case that an edge $e$ is reserved for $P$ although $P$ can not be merged via $e$ to another path. However, by construction, in this case, the endpoints of $P$ are not notified about the reservation of $e$.

By Step~\ref{step:pick}, each path in $\mathcal{P}_i$ adds to $M_i$ at most one merging operation (which corresponds to one of the reservations it received).  

Specifically, for each $P\in \mathcal{P}_i$, after receiving the reservations, the endpoints of $P$ send up the tree $T(P)$ one of the reservations they received. The leader then picks one reservation and sends its details down the tree $T(P)$.
{\color{black}
\subsection{Putting Things Together}
Lemma~\ref{lemma:numofiter} states that a Hamiltonian path is found within $O(\log n)$ rounds w.h.p. We showed in Appendix~\ref{sec:imp} that our algorithm can be implemented in the \congest\ model. To get a Hamiltonian cycle from the found Hamiltonian path, we observe that the endpoints of the found Hamiltonian path, $u$, and $v$, satisfy that $d(u)+d(v) \geq n$ (this is true for both \dirac\ and \ore\ graphs). Claim~\ref{lemma:basic} implies that this computed Hamiltonian path is cycled. In turn, it can be transformed into a Hamiltonian cycle in a constant number of rounds. This observation is formalized in the following claim.}
\begin{claim}\label{claim:oneround}
    Let $P=(u, \ldots, v)$ denote a Hamiltonian path computed by Algorithm~\ref{alg:hamalg}. If $d(u)+d(v) \geq n$, then $P$ can be transformed into a Hamiltonian cycle in a constant number of rounds in the \congest\ model. 
\end{claim}
{\color{black} This concludes the proof of Theorem~\ref{thm:ham}. }
\section{Rahman-Kaykobad (\rk) Graphs}\label{sec:rk}
 A graph family that generalizes Ore graphs was introduced by Rahman and Kaykobad~\cite{RK05}, which we refer to as \emph{\rk\ graphs}. \rk\ graphs are defined as follows.
 
 \begin{definition}[\cite{RK05}]\label{def:RK}
    Let $G = (V,E)$ be a connected graph with $n$ vertices. We say that $G$ is an \rk\ graph if for all pairs of distinct nonadjacent vertices $u,v\in V$ it holds that 
        $d(u) + d(v) +\delta(u,v) \geq n + 1$. 
 \end{definition}
 
For every \rk\ graph, Rahman and Kaykobad~\cite{RK05} showed that a Hamiltonian path exists. 
 \begin{theorem}[Thm.~1.6~\cite{RK05}, restated]
    Let $G = (V,E)$ be an \rk\ graph. Then $G$ has a Hamiltonian path.
 \end{theorem}
 
 \sloppy In this section, we design and analyze a distributed  \congest\ algorithm that finds a Hamiltonian path in \rk\  graphs. 
 We start by showing several structural properties of \rk\  graphs and then describe how to adapt Algorithm~\ref{alg:hamalg} so that it finds a Hamiltonian path also in \rk\ graphs. {\color{black}We note that similarly to \dirac\ graphs, in case the output of the algorithm for a given instance satisfies the condition of Claim~\ref{claim:oneround} then the resulting path can be turned into a Hamiltonian cycle. This is always the case for \ore\ graphs, as formalized in the following theorem.
 \mainhamtwo*
 }
 
\subsection{Structural Properties of \rk\ graphs}
Throughout this section we assume w.l.o.g. that there is a vertex $v^* \in V$ such that $d(v^*) < \frac{n}{2}$ (otherwise, this is a \dirac\ graph). 

\paragraph*{Notation. }
Let $G = (V,E)$ be an \rk\ graph. 
Let $v^* \in V$ be such that $d(v^*) < \frac{n}{2}$.
We define $A \triangleq \{u \in V \mid \delta(u,v^*)=1\}$, $B \triangleq \{u \in V \mid \delta(u,v^*)=2\}$, $C \triangleq \{u \in V \mid \delta(u,v^*)=3\}$, $D \triangleq \{u \in V \mid \delta(u,v^*)=4\}$, and  $E \triangleq \{u \in V \mid \delta(u,v^*)=5\}$. Let $H \triangleq \{u \in V \mid d(u) \geq \frac{n}{2}\}$. We refer to the vertices in $H$ as ``heavy''. Let $\hat{A} \triangleq A \setminus H$, $\hat{B} \triangleq B \setminus H$, and $\hat{C} \triangleq C \setminus H$. Let $\hat{A}_+ \triangleq \hat{A}\cup\{v^*\}$. 

\begin{claim}
    The set $\hat{B} \subset V$ is empty, i.e., all vertices in $B$ are in $H$.
\end{claim}
\begin{proof}
    Let $u \in B$. By Definition~\ref{def:RK} and the definition of the set $B$ it is implied that $d(v^*)+d(u) \geq n-1$. Since $d(v^*) < \frac{n}{2}$, it follows that $d(u) \geq n-1 - d(v^*) > n - 1 - \frac{n}{2} = \frac{n}{2}-1$, i.e., $d(u) \geq \frac{n}{2}$, as required. 
\end{proof}

In what comes next, we prove that if for a given input \rk\ graph it holds that $D \neq \emptyset$, then one can compute a Hamiltonian Path in $\Theta(1)$ rounds. 
\begin{claim}\label{claim:E}
    The set $E$ is empty.
\end{claim}
\begin{proof}
    Let $u \in C$. By Definition~\ref{def:RK} and the definition of the set $C$ it is implied that $d(v^*)+d(u) \geq n+1-3= n-2$. 
    Since $N(v^*) \cap N(u) = \emptyset$, $v^*  \not\in N(v^*)$, $u \not\in \cup N(u)$, and since $N(v^*) = A$, it follows that $N(u) = B \cupdot (C\setminus\{u\}) \cupdot D$. In turn, this means that there are no vertices of distance more than $4$ from $v^*$, as required. 
\end{proof}

\begin{claim}\label{claim:CD}
    If $D \neq \emptyset$, then 
    \begin{itemize}
        \item $|B|=1$,
        \item $D = k_{|D|}$, and
        \item The subgraph induced by $C$ and $D$ is a $k_{|C|,|D|}$. 
    \end{itemize}
\end{claim}
\begin{proof}
    Let $u \in D$. By Definition~\ref{def:RK} and the definition of the set $D$ it is implied that $d(v^*)+d(u) \geq n+1-4= n-3$. 
    Since $N(v^*) \cap N(u) = \emptyset$, $v^*  \not\in N(v^*)$, $u \not\in \cup N(u)$, $N(v^*) = A$, and due to Claim~\ref{claim:E} (i.e., $E = \emptyset$), it follows that $N(u) = C \cupdot (D\setminus\{u\})$. In turn, this means that a single vertex is not in $N(v^*) \cupdot N(D)$, that is, $|B|=1$, as required. 
\end{proof}

\begin{claim}
    If $D \neq \emptyset$, then 
    \begin{itemize}
        \item $A = k_{|A|}$, and
        \item The subgraph induced by $A$ and $B$ is a $k_{|A|,1}$. 
    \end{itemize}
\end{claim}
\begin{proof} 
    Let $u \in A$ and let $v \in D$. By Definition~\ref{def:RK} and the definition of the sets $A$ and $D$, it is implied that $d(u)+d(v) \geq n+1-3= n-2$. 
    Since $N(u) \cap N(v) = \emptyset$, $u  \not\in N(u)$, $v \not\in \cup N(v)$, $N(v^*) = A$, and due to Claim~\ref{claim:CD} (i.e., The subgraph induced by $C$ and $D$ is a $k_{|C|,|D|}$.), it follows that $N(u) = \{v^*\} \cupdot (A\setminus\{u\}) \cupdot B$. This means every vertex in $A$ is adjacent to all the other vertices in $A$, $B$, and $v^*$, as required. 
\end{proof}

Due to the above claims, there is a $\Theta(1)$ rounds \congest\ algorithm that decides whether the input graph satisfies that $D \neq \emptyset$: it simply performs a BFS from $v^*$. If this BFS finds vertices of distance four from $v^*$, then the algorithm computes a Hamiltonian path for each clique separately and then ``stitches'' these paths to a single one. Hence, this simple case can be entirely dealt with a $\Theta(1)$ rounds \congest\ algorithm. From this point until the end of this section, we assume that the input graph satisfies that $D = \emptyset$.

In the following claim, we argue that the subgraph induced on $C$ is a clique, and the subgraph induced on $B \cup C$ contains a fully bipartite graph between $B$ and $C$. 
\begin{claim}\label{claim:biclique}
 For each $c \in C$ it holds that $N(c) = (B \cup C) \setminus \{c\} $.
\end{claim}
\begin{proof}
     By Definition~\ref{def:RK} and the definition of the sets $A$ and $C$ it is implied that $d(v^*)+d(c) \geq n+1-3=n-2$. Since $|A| = d(v^*)$, it follows that $d(c) \geq n-2 - |A|$. On the other hand, the vertex $c$ is not a neighbor of itself, nor of $v^*$, hence the degree of $c$ can be bounded from above, as follows
     $d(c) \leq n - 1 - 1 - |A| = n-2 - |A|$. Hence, $d(c) = n-2 - |A|$. Put differently, the neighbors of each vertex $c \in C$ are all of the rest of the vertices in $V$ which are not in $A$, not $v^*$, and not $c$ itself, i.e., the neighbors of $c$ are all the vertices in $(B \cup C) \setminus \{c\}$, as required. 
\end{proof}

A corollary of Claim~\ref{claim:biclique} is that the vertex set $V$ of an \rk\ graph $G=(V,E)$ can be partitioned to the corresponding $A,B$, and $C$ sets. 
\begin{corollary}
     $V=\{v^*\}\cupdot A\cupdot B \cupdot C$. 
\end{corollary}

\begin{claim}\label{claim:Aclique}
    The subgraph induced on $\hat{A}_+$ is a clique. 
\end{claim}
\begin{proof}
    Let us assume towards a contradiction that there are two vertices $a_1, a_2 \in \hat{A}_+$ such that $\delta(a_1, a_2)=2$ (obviously, $\delta(a_1, a_2)\leq 2$, since $v^*$ is connected to all the vertices in $A$). Definition~\ref{def:RK} implies that $d(a_1)+d(a_2) \geq n+1-2=n-1$. Hence, it cannot be that both $a_1$ and $a_2$ have degrees that are strictly less than  $\frac{n}{2}$, as required. 
\end{proof}

\subsection{Spanning sub-graphs}

Before we describe how Algorithm~\ref{alg:hamalg} can be adapted to work for \rk\ graphs, we first describe how we modify the spanning trees (see Section~\ref{sec:spanningtrees}) to support paths in which not all vertices have a degree of at least $n/2$.

By the proof of Claim~\ref{clm.clm15} for any pair of vertices $u$, $v$ in $H$ it holds that $\delta(u, v) \leq 2$. Therefore it is possible to maintain for each path $P$ in the path cover a tree that spans the vertices in $P \cap H$ while each edge is participating in at most $2$ different trees.

In addition, Claim~\ref{claim:Aclique} and Claim~\ref{claim:biclique} imply that there is a spanning tree of depth $1$ for the vertices in $P \cap \hat{A}_+$ and a spanning tree of depth $1$ for the vertices in $P \cap \hat{C}$.

Therefore, the merges can be carried out just as before while having a constant factor blow-up in the round complexity. Specifically, the updates are performed in three phases, where there is a phase for the vertices of each one of the sets: $H$, $\hat{A}_+$, and $\hat{C}$.   

Regarding the global spanning tree, it suffices that it spans just the vertices in $H$. This is because it is used for pairing paths in which both endpoints are in $H$. Therefore, its construction remains unchanged.

\subsection{Adaptation of Algorithm~\ref{alg:hamalg} to \rk\ graphs}

In this section, we describe how Algorithm~\ref{alg:hamalg} can be adapted to 
work under the promise that the input graph satisfies the \rk\ condition.

As in Algorithm~\ref{alg:hamalg}, let $\mathcal{P}_i$ denote the path cover of the graph at the beginning of iteration  $i$.
Let $\mathcal{H}_i$ denote the paths in $\mathcal{P}_i$ for which both endpoints are in $H$.  
If $\mathcal{P}_i = \mathcal{H}_i$, then in iteration $i$, we can proceed as we did in Algorithm~\ref{alg:hamalg}.~\footnote{The correctness of this case follows from the proof of Lemma~\ref{lemma:leftdeg}. } 

Consequently, at the beginning of each iteration, our goal is to first handle the paths in $\mathcal{P}_i$ which are not in $\mathcal{H}_i$ and to get rid of them.
We merge these paths into paths where both endpoints are in $H$ (while leaving at most $2$ paths that do not have this property).

These merging operations are quite straightforward and are performed in two stages. In the first stage, we deal with paths with an endpoint in $\hat{A}_+$, as follows.  Given a list of endpoints in $\hat{A}_+$ and the ids of their corresponding paths, one can merge these paths greedily and obtain paths in which both endpoints are in $H\cup \hat{C}$ and at most a single path for which this property does not hold (namely, it has at least one endpoint is in $\hat{A}_+$).
In the second stage, we repeat the same merging process with respect to paths with endpoints in $\hat{C}$ and obtain a path cover in which for all paths, except for at most $2$ paths (i.e., one with an endpoint in $\hat{A}_+$ and one with an endpoint in $\hat{C}$), both endpoints are in $H$.

Since both $\hat{A}_+$ and $\hat{C}$ are cliques, it is not hard to see that these merges can be carried out in a constant number of rounds. 

\paragraph*{Termination of the algorithm.}
It remains to deal with the cases in which the path cover consists of two paths, one with an endpoint in $\hat{A}_+$ and one with an endpoint in $\hat{C}$.
Let $P_1$ and $P_2$ denote these paths, respectively, and $x_1$ and $x_2$ denote the respective endpoints.

The first case is when the other endpoint of $P_2$, denoted by $y_2$, is in $\hat{A}_+$. Since the subgraph induced on $\hat{A}_+$ forms a clique, in this case, we can merge $P_1$ and $P_2$ by a concatenation merge.

The second case is when the other endpoint of $P_1$, $y_1$, is in $B\cup C$. In this case, we can also concatenate $P_1$ and $P_2$ (by connecting $x_2$ to $y_1$).

Otherwise, we show that either we can merge $P_1$ and $P_2$ with an elementary merge or the subgraph induced on the vertices of $P_1$ contains a Hamiltonian cycle and likewise for $P_2$. Since the graph is connected, these cycles can be merged into a single Hamiltonian path.

We first consider $P_1$. If $y_1\in \hat{A}_+$, then we are done. Otherwise, it follows that $y_1\in A\cap H$. Therefore it holds that $\delta(x_1, y_1) = 2$, and so the sum of their degrees is at least $n-1$. It follows that either $P_1$ is cycled or there is an edge $e$ in $P_2$ such that $P_1$ can be merged to $P_2$ via $e$.

Similarly, in $P_2$, if $y_2 \in B\cup C$ then we are done. Otherwise, it follows that $y_2 \in A\cap H$. Therefore $\delta(x_2, y_2) = 2$, which implies that either $P_2$ is cycled or that there is an edge $e$ in $P_1$ such that $P_2$ can be merged to $P_1$ via $e$.

This completes the case analysis of the (end) cases of the termination of the algorithm. 

\end{document}